\documentclass{article}

\usepackage{arxiv}

\usepackage[utf8]{inputenc} 
\usepackage[T1]{fontenc}    
\usepackage{hyperref}       
\usepackage{url}            
\usepackage{booktabs}       
\usepackage{amsfonts}       
\usepackage{nicefrac}       
\usepackage{microtype}      
\usepackage{lipsum}
\usepackage{graphicx}
\usepackage{amsmath}
\usepackage{amssymb}
\usepackage{amsthm}
\usepackage{natbib}
\usepackage{comment}
\graphicspath{ {./images/} }

\newtheorem{lemma}{Lemma}
\newtheorem{remark}{Remark}

\newtheorem{assumption}{Assumption}
\newtheorem{proposition}{Proposition}
\newtheorem{theorem}{Theorem}
\newtheorem{corollary}{Corollary}
\title{Estimating High-Dimensional Discrete Choice Model of Differentiated Products with Random Coefficients}

\author{
  Masayuki Sawada\\
  Institute of Economic Research,\\
  Hitotsubashi University\\
  \texttt{m-sawada@ier.hit-u.ac.jp} \\
   \And
  Kohei Kawaguchi \\
  Department of Economics, \\
  Hong Kong University of Science and Technology\\
  \texttt{kkawaguchi@ust.hk} \\
}

\begin{document}
\maketitle
\begin{abstract}
We propose an estimation procedure for discrete choice models of differentiated products with possibly high-dimensional product attributes. In our model, high-dimensional attributes can be determinants of both mean and variance of the indirect utility of a product. The key restriction in our model is that the high-dimensional attributes affect the variance of indirect utilities only through finitely many indices. In a framework of the random-coefficients logit model, we show a bound on the error rate of a $l_1$-regularized minimum distance estimator and prove the asymptotic linearity of the de-biased estimator.
\end{abstract}


\section{Introduction}

There are many occasions in which high-dimensional product attributes are available for estimating and predicting demands for differentiated products, especially in context where machine learning techniques such as pattern recognition and natural language processing can generate high-dimensional representation of product attributes. For example, we can consider a consumer choice over clothes. In addition to the typical characteristics of prices, countries of origin, and materials, there are numerous varieties of design patterns that can influence a consumer's choice. The choice over books is dependent on their contents. With a natural language processing technique, we can potentially represent their contents with a high-dimensional semantic vectors. In this paper, we investigate a framework that enables us to integrate these potentially informative but not-yet-fully-used product attributes in estimation, inference, and prediction of demands for differentiated products.

We are not the first considering the estimation of high-dimensional discrete choice model. \cite{gillenBLP2LASSOAggregateDiscrete2019} studied an estimation of a model that extends \citet[henceforth BLP]{berryAutomobilePricesMarket1995}'s random-coefficients discrete-choice model with high-dimensional attributes and applied to the study of political campaign for elections. In their model, the product attributes can grow exponentially relative to the sample size; however, random coefficient is only allowed for the price. Under this assumption, the high-dimensional attributes can only affect the mean indirect utility. Thus, they can apply a $l_1$-regularized least squares method to the inverted mean indirect utility to select relevant variables.

Recent methodological developments in statistics and econometrics allows us to apply a high-dimensional estimation procedure to a broader class of models. In particular, the recent work of \cite{belloniHighdimensionalEconometricsRegularized2018} provides a set of useful results for $l_1$-regularized minimum distance estimation following the techniques developed by \cite{frankStatisticalViewChemometrics1993} and \cite{tibshiraniRegressionShrinkageSelection1996}. We derive the property of our $l_1$-regularaized minimum distance estimator and its debiased version based on their results. Because the objective function of the BLP model is non-linear in the parameters, we use the contraction inequality theorem of \cite{ledouxProbabilityBanachSpaces1991} to control the tail probability of the estimation error. This contraction inequality exploits the Lipschitz continuity of the objective function with respect to a single index. In the context of the BLP discrete choice model with non-linear but smooth moment conditions, we show that the analogue principle can be applied to the case with multiple indices. This allows us to have fixed number of random coefficients on the indices of potentially high-dimensional attributes. 

The rest of the paper proceeds in the following manner. In the next section, we describe our model and introduce the regularized GMM (RGMM) problem for our model. In section 3, we show the probability bound for the estimation error from the regularized GMM problem. In section 4, we consider the de-biased procedure for the proper inference. The last section concludes.

\section{Demand estimation as a regularlized GMM}
\subsection{Model}

Consider there are $J$ products in each market $i \in \{1,\ldots,n\}$. Let us denote $[\cdot]$ for an integer indicates index set $\{1,\ldots, \cdot\}$.
Each product $j$ shares a non-zero demand $S_{ij} \in (0,1)$ in each market $i$. Each product $j$ in a market $i$ has observed $L$ attributes $\{x_{ijl}\}_{l \in [L]}$ including cost of attaining the product $j$, $-p_{ij}$, and an unobserved attribute $\xi_{ij}$. 

In this paper, we consider the high-dimensionality in the product attributes $x_{ijl}$. 
The key restriction is that we assume there are $G$ known finite partitions of the high-dimensional product characteristics $[L]$. In particular, we consider the following indirect utility for a product $j$ in a market $i$:
\[
 u_{ij} = \sum_{g \in [G]} \sum_{l \in L_g} x_{ijl}(\beta_{l} + \gamma_l \tilde{\beta}_{g}) + \xi_{ij} + \epsilon_{ij}
\]
where $\epsilon_{ij}$ is an idiosyncratic error term and $L_g$ is mutually exclusive subset of $[L]$ for each $g \in [G]$ such that $\cup_{g \in [G]} [L_g] = [L]$. Note that $\sum_{g \in [G]} \sum_{l \in L_g} x_{ijl}\beta_{l} + \xi_{ij} = x_{ij}'\beta + \xi_{ij}$ with $\beta = (\beta_1,\ldots,\beta_L)'$ captures the mean utility of the product $j$ in the market $i$. Also, $\sum_{g \in [G]} \sum_{l \in L_g} x_{ijl}\gamma_{l} \tilde{\beta}_{g} = \sum_{g \in [G]} x_{ijg}'\gamma_g \tilde{\beta}_{g}$ with $x_{ijg} = [x_{ijl}]_{l \in L_g}'$ and  $\gamma_g = [\gamma_l]_{l \in L_g}'$ and $\tilde{\beta}_{g} \sim_{iid} N(0,1)$ capture the individual heterogeneity in the utility for the product $j$ in the market $i$. We assume $\tilde{\beta}_{g}$ is independent of $\tilde{\beta}_{g'}$, and the group specific variance of the random coefficient $\tilde{\beta}_g$ is normalized to 1 for each corresponding vector $\gamma_g$. Therefore, the variance of the indirect utility of the product $j$ in the market $i$ is $\sum_{g \in [G]} (\sum_{l \in [L_g]} x_{ijl} \gamma_l)^2$. This finite indices restriction allows us to apply the contraction inequality theorem, \cite{ledouxProbabilityBanachSpaces1991}, which is in principle applied to a Lipschitz transformation of a single index. If the random coefficients are not restricted, then the number of indices grows to infinity as the number of attributes grows infinity, and we cannot apply the contraction inequality to bound the estimation error. For the reminder of the paper, let $g(l)$ represents the group $g \in [G]$ that the attribute $l$ belongs to.

The mean utility is the same as usual differentiated product demand model such as the BLP model. The heterogeneity term is different from the usual random coefficient model. This model can be seen as a special case of the usual random coefficient model that the product attributes $x_{ijl}$ and $x_{ijl'}$ share the same individual preference shock $\beta_{g}$ if $l, l' \in L_g$. In other words, consumers observe the set of characteristics $\{x_{ijl}\}_{l \in L_g}$ as a common index characteristics of $x_{ijg}'\gamma_g$, but individuals may have different preference over the index $x_{ijg}'\gamma_g$. In the example of the design pattern, we share the same objective descriptions of the design, but the subjective preference over the descriptions differ across people. 
An important feature is that we do not restrict $\beta_l = \gamma_l$. Therefore, the mean $\sum_{l \in L_g} x_{ijl}\beta_l$ and the variance $(\sum_{l \in L_g} x_{ijl}\gamma_l)^2$ of each utility load from a group of characteristics $L_g$ may be unrelated each other.

\subsection{Moment condition}

Suppose that $\epsilon_{ij}$ follows iid Type-I extreme value distribution, then the parameters $\theta \equiv \{\beta',\gamma'\}'$ and observed characteristics $x_{ij}$ pin down the share of the product $j$ in the market $i$ as the function of mean utility vector $x_{i}'\beta + \xi_i$ such that
\[
 s_j(x_{i},x_{i}'\beta + \xi_{i};\theta) = \int_{\tilde{\beta}} \frac{\exp(x_{ij}'\beta + \xi_{ij} + \sum_g x_{ijg}'\gamma_g \tilde{\beta}_{g})}{1 + \sum_{j' \in [J]} \exp(x_{ij'}'\beta + \xi_{ij'} + \sum_g x_{ij'g}'\gamma_g\tilde{\beta}_{g})} dF_{\tilde{\beta}} = S_{ij}
\]
where $F \sim N_{[G]}(0,I_{[G]})$ and $S_{ij}$ denotes the population share of the product $j$ in the market $i$. \footnote{To simplify the argument, we ignore the measurement error issue of the observed market share from the population share for now.}

Here we assume that the unobserved product type $\xi_{ij}$ is mean independent of a vector of instruments $w_{ij}$,
\[
 E[\xi_{ij}|w_{ij}] = 0
\]
where the expectation is taken over the markets $i$ for each $j \in [J]$. This conditional moment restriction leads to a set of unconditional moment conditions with transformed vector of $K$ instruments 
for each product $j$, $h_{jk}(w_{ij}), k \in [K]$, such that
\[
 E[\xi_{ij} h_{jk}(w_{ij})] = 0
\]
for each $j \in [J]$. 
\cite{berryEstimatingDiscreteChoiceModels1994} shows that there exists unique inverse functions of the share $s_j(x,\cdot;\theta)$, such that 
\begin{equation}
 s_j(x_i,s^{-1}(x_i,S_i;\theta);\theta) = S_{ij}. \label{eq.share}
\end{equation}
Therefore, the moment condition is now
\[
 E[(s_j^{-1}(x_{i},S_i;\theta) - x_{i}'\beta)h_{jk}(w_{ij})] = 0, \forall j \in [J], k \in [K].
\]
Now let 
\[
 \xi_{j}(\tilde{x};\theta) \equiv s_j^{-1}(x,S;\theta) - x_{ij}'\beta
\]
where $\tilde{x} \equiv (x',S',w')'$.

\subsection{Regularized GMM problem}
Following \cite{belloniHighdimensionalEconometricsRegularized2018}, we consider a regularized GMM approach for the moment condition above. In particular, let $f(\tilde{x};\theta)$ be the score function vector of $(J,K)$ elements with
\[
 f_{jk}(\tilde{X};\theta) \equiv \xi_{j}(\tilde{X};\theta) h_{jk}(W),
\]
for each $j \in [J]$ and $k \in [K]$ where $\tilde{X} \equiv (X',S',W)'$, and let
\[
 f(\theta) \equiv A E[f(\tilde{X};\theta)]
\]
and
\[
 \hat{f}(\theta) \equiv \hat{A} \mathbb{E}_n [f(\tilde{X};\theta)]
\]
with some weight matrix $A$ and its estimate $\hat{A}$, where $\mathbb{E}_n[\cdot]$ represents sample mean of a random vector. For now, let $A = \hat{A} = I$.

The regularized GMM estimator $\hat{\theta}$ solves the following optimization problem
\[
 \min_{\theta \in \Theta} \|\theta\|_1 : \|\hat{f}(\theta)\|_{\infty} \leq \lambda
\]
for some regularization parameter $\lambda$.

\section{Bounds on the estimation error}

\cite{belloniHighdimensionalEconometricsRegularized2018} show the rate of convergence for the estimation error under two conditions in addition to the regularization condition which is the constraint of the optimization problem shown above. Below we cite their statement under three high-level conditions
\begin{proposition}[Proposition 3.1 of \cite{belloniHighdimensionalEconometricsRegularized2018}]
 Assume the following three conditions:
 \begin{enumerate}
     \item (Regularization) The regularization parameter $\lambda$ satisfies
      \[
      \|\hat{f}(\theta_0)\|_{\infty} \leq \lambda
     \]
     with probability at least $1 - \alpha$
     \item (Identifiability) The population moment function satisfies the following:
     \[
      \left\{\|f(\theta) - f(\theta_0)\|_{\infty} \leq \epsilon, \theta \in \mathcal{R}(\theta_0) \right\}
     \]
     implies
     \[
      \|\theta - \theta_0\|_l \leq r(\epsilon;\theta_0,l)
     \]
     for all $\epsilon > 0$ where $\mathcal{R}(\theta_0) \equiv \{\theta \in \Theta: \|\theta\|_1 \leq \|\theta_0\|_1\}$, and $r(\cdot;\theta_0,l)$ is a weakly increasing rate function depending on the semi-norm $l$.
     \item (Empirical moment restriction) The empirical moment function satisfies
     \[
      \sup_{\theta \in \mathcal{R}(\theta_0)} \|\hat{f}(\theta) - f(\theta)\|_{\infty} \leq \epsilon_n
     \]
     with probability at least $1 - \delta_n$.
 \end{enumerate}
Then with probability at least $1 - \alpha - \delta_n$,
\[
 \|\hat{\theta} - \theta_0\|_l \leq r(\lambda + \epsilon_n;\theta_0,l).
\]
\end{proposition}

\subsection{Identifiability condition}
For the second condition of the identifiability, \cite{belloniHighdimensionalEconometricsRegularized2018} offers the following sufficient condition
\begin{assumption}[Condition NLID for exactly sparse parameters] \label{ass.nlid}
 Suppose that there exist $T \subset [L]$ with cardinality $s$ such that $\theta_{l0} \neq 0$ only for $l \in T$.

 For each $q \in \{1,2\}$, suppose that there exists a sequence $\mu_n$ such that
 \[
  k(\theta_0,l_q) \equiv \inf_{\theta \in \mathcal{R}(\theta_0):\|\theta - \theta_0\|_q > 0}\|G(\theta - \theta_0)\|_{\infty}/\|\theta - \theta_0\|_q \geq s^{-1/q}\mu_n
 \]
 where $G$ is the Jacobian matrix of $f(\theta)$.
 
 Suppose further that
 \[
  \{\|f(\theta) - f(\theta_0)\|_{\infty} \leq \epsilon, \theta \in \mathcal{R}(\theta_0)\}
 \]
 implies that 
 \[
  \|G(\theta - \theta_0)\|_{\infty}/2 \leq \epsilon
 \]
 for all $\epsilon \leq \epsilon^*$ for some $\epsilon^*$.
\end{assumption}

The last condition in the above assumption \ref{ass.nlid} is specific to the non-linear problem. Nevertheless, this assumption does not bind in our model because our target moment function is continuously differentiable everywhere.

The second condition in assumption \ref{ass.nlid} regulates the modulus of continuity $k(\theta_0,l)$. Lemma 3.1 of Belloni et al. (2018) offers a sufficient condition for the second condition in exactly sparse model. For the linear IV regression model, for any sub-vector of covariates $X$, we need some sub-vector of instruments $W$ such that $E[W'X]$ is non-singular. In other words, there exists some instruments that are strong for any sub-vector of endogenous covariates. In our context of the BLP model, the Jacobian matrix is $JK \times 2L$ matrix with each $l$ entry for $jk$ element as
\[
    G_{jk,l}(\tilde{X},\theta) =
    \begin{cases}
     E[h_{jk}(W)X_{jl}] & \mbox{ for } 1 \leq l \leq L\\
     E[h_{jk}(W)D_j(\theta) \int \tilde{\beta}_{g(l)} s(\tilde{\beta};\tilde{X},\theta) \sum_{j'=1}^{J} (1 - s_{j'}(\tilde{\beta};\tilde{X},\theta))X_{j'l}dF_{\tilde{\beta}}] & \mbox{ for } L+1 \leq l \leq 2L
    \end{cases}
\]
where $s(\tilde{\beta};\tilde{X},\theta)$ is $J\times 1$ vector of $s_j(\tilde{\beta};\tilde{X},\theta) \equiv \frac{\exp(X'_j \beta + \xi_j(\tilde{X};\theta) + \sum_g X_{jg}'\gamma_g \tilde{\beta}_{g})}{1 + \sum_{j' \in [J]} \exp(X'_{j'} \beta + \xi_{j'}(\tilde{X};\theta) + \sum_g X_{j'g}'\gamma_g\tilde{\beta}_{g})}$ and $D_{j}(\theta)$ is the $j$-th row of the inverse matrix of the Jacobian matrix of $s(\tilde{X};\theta) \equiv s(X,X'\beta + \xi(\tilde{X};\theta);\theta)$ vector with respect to the mean utility vector. Therefore, the modulus of continuity condition requires that the variables $h_{jk}(W)$ serve as the strong instruments for the attribute $l$ of the product $j$, $X_{jl}$, as well as the weighted sum of the attributes $l$ of the products $j'$ across the market $X_{j'l}$. This is not a strong restriction for the most of the attributes as we often assume the attributes are exogenous. For the endogenous attributes, we need to be cautious on the restriction as the instruments are not necessarily strong in particular when the asymptotic is considered for the size of markets $J$ rather than the number of markets $n$. See \cite{armstrongLargeMarketAsymptotics2016} for the relevant discussion.

\begin{lemma}[Lemma 3.4 of \cite{belloniHighdimensionalEconometricsRegularized2018}]
 Under assumption \ref{ass.nlid}, for all $0 < \epsilon \leq \epsilon^*$
 \[
  \left\{\|f(\theta) - f(\theta_0)\|_{\infty} \leq \epsilon, \theta \in \mathcal{R}(\theta_0) \right\}
 \]
 implies
 \[
  \|\theta - \theta_0\|_l \leq r(\epsilon;\theta_0,l) \leq 2\epsilon s^{1/q} \mu_n^{-1}.
 \]
\end{lemma}

\subsection{Tail probability bound}

Now, let $\Delta \theta \equiv \theta - \theta_0$ for any $\theta, \theta_0 \in \Theta$, where $\Theta$ is a subset of $\mathbb{R}^{2L}$ defined such that $\Delta \theta \in \Theta$.

Next, we consider the bound for the tail-probability of the estimation error process:
\[
 \sup_{\Delta \theta \in \Theta, j \in [J], k \in [K]} \left| \mathbb{G}_n (f_{jk}(\tilde{X};\theta_0 + \Delta \theta) - f_{jk}(\tilde{X};\theta_0)) \right|
\]
where $\mathbb{G}_n X$ is the empirical process of the sequence $\{X_i\}_{i \in [n]}$.

For the rest of the discussion, we introduce the following indices. For each $g \in \{0,1,\ldots, G\}$ and $j \in [J]$, let $X_{j,g}$ be a sub-vector of $X_j$ corresponding $g$-th partition of $[L]$ with $X_{j,0} \equiv X_j$, and let $\theta_g$ and $\nu_{jg}$ be defined as
 \[
  \nu_{jg}\equiv X_{j,g}\theta_{g} \equiv
  \begin{cases}
    X_{j}'\beta & \mbox{ if } g = 0\\
    X_{j,g}'\gamma_{g} & \mbox{ if } g > 0
  \end{cases}
 \]
 and
  \[
  \xi_{j}(\nu;\tilde{X}) \equiv \xi(\tilde{X};\theta)
 \]
 such that $\nu_{j'g} = X_{j',g}\theta_g$ for every $j' \in [J]$ and $g \in \{0,1,\ldots, G\}$.
 
 Also, let
 \[
 \nu_{j'g0}\equiv X_{j',g}\theta_{g0} \equiv
  \begin{cases}
    X_{j'}\beta_0 & \mbox{ if } g = 0\\
    X_{j',g}\gamma_{g0} & \mbox{ if } g > 0.
  \end{cases}
 \]
 In lemma \ref{lemma.Lip} in the Appendix, we show that the score functions $f_{jk}$ are Lipschitz continuous in $\nu_{j'g}$ uniformly for every $\nu_{-j',-g}$ with the Lipschitz constant $J$ times some universal constant. Using this property, we employ the Ledaux-Talagrand contraction inequality as follows:
\begin{theorem} \label{thm.tailbound}
 In addition to the assumptions for lemma \ref{lemma.Lip} in the Appendix, suppose the following
 \begin{enumerate}
     \item 
        $\sup_{\Delta \theta \in \Theta, j \in [J],k \in [K]} \mathbb{E}_n Var(f_{jk}(\tilde{X},\theta_0 + \Delta \theta) - f_{jk}(\tilde{X},\theta_0)) \leq B_{1n}^2$, and
    \item $\max_{j \in [J], l \in [L], k \in [K]} \mathbb{E}_n (X^2_{jl} h^2_{jk}(W)) \leq B_{2n}^2$, and
        $\|n^{-1/2} \mathbb{G}_n(f(\tilde{X},\theta_0))\|_{\infty} \leq n^{-1/2} l_n$ with probability at least $1 - \delta_n/6$
 \end{enumerate}
 then,
 \[
  \sup_{\theta \in \mathcal{R}(\theta_0)} \|\hat{f}(\theta) - f(\theta)\|_{\infty} \leq n^{-1/2} (\tilde{l}_n + l_n)
 \]
 with probability at least $1 - \delta_n$, where
 \[
  \tilde{l}_n \equiv C(B_{1n} + (J^2G)(2\sqrt{2}B_{2n}\sup_{\theta}\|\theta\|_1 \log^{1/2}(8J^2 GKL/\delta_n)))
 \]
 with a universal constant $C$.
\end{theorem}
\begin{proof}
 In the same argument of theorem 3.2 of \cite{belloniHighdimensionalEconometricsRegularized2018}, $\|n^{-1/2} \mathbb{G}_n(f(\tilde{X},\theta_0))\|_{\infty} \leq n^{-1/2} l_n$ implies that we only need to bound the following empirical process
 \[
 \max_{j \in [J],k \in [K]} \sup_{\Delta \theta} \left|\mathbb{G}_n (f_{jk}(\tilde{X};\theta_0 + \Delta \theta) - f_{jk}(\tilde{X};\theta_0)) \right|.
 \]

 By taking $t^2 \geq 16 B_{1n}^2$, we may apply Chebyshev inequality and symmetrization lemma (Lemma 2.3.7 of \citealp{vandervaartWeakConvergenceEmpirical1996}) so that
 \begin{align*}
   P&\left(\max_{j \in [J], k \in [K]} \sup_{\Delta \theta \in \Theta} \left|\mathbb{G}_n (f_{jk}(\tilde{X};\theta_0 + \Delta \theta) - f_{jk}(\tilde{X};\theta_0)) \right| > t\right)\\
   &\leq 4 P\left(\max_{j \in [J], k \in [K]} \sup_{\Delta \theta \in \Theta} \left|\mathbb{G}_n \sigma (f_{jk}(\tilde{X};\theta_0 + \Delta \theta) - f_{jk}(\tilde{X};\theta_0)) \right| > t/4\right)  
 \end{align*}
 where $\sigma$ is iid Rademacher variable taking $-1$ and $1$ with equal probability independent of all the others. 
 
 Following the step 1 of lemma D.3 of \cite{belloniHighdimensionalEconometricsRegularized2018}, by conditioning on $\Omega_n \equiv \{\max_{j \in [J],l \in [L],k \in [K]} E_n(X_{jl}^2 h_{jk}^2(W)) \leq B_{2n}^2\}$, we consider bounding the tail probability conditional on the event $\Omega$ and $\tilde{X}$,
 \begin{align*}
  P&\left(\max_{j \in [J], k \in [K]} \sup_{\Delta \theta \in \Theta}\left|\mathbb{G}_n \sigma
  (f_{jk}(\tilde{X};\theta_0 + \Delta \theta) - f_{jk}(\tilde{X};\theta_0)) \right| > t/4\middle|\Omega_n,\tilde{X}\right).     
 \end{align*}
 From now on, omit the conditioning for the notational simplicity.
 By Markov inequality, we have
 \begin{align*}
  P&\left(\max_{j \in [J], k \in [K]} \sup_{\Delta \theta \in \Theta}\left|\mathbb{G}_n \sigma (f_{jk}(\tilde{X};\theta_0 + \Delta \theta) - f_{jk}(\tilde{X};\theta_0)) \right| > t/4\right) \\
  &\leq \frac{E_{\sigma}\exp\left(\phi \max_{j \in [J], k \in [K]} \sup_{\Delta \theta \in \Theta} \left|\mathbb{G}_n \sigma (f_{jk}(\tilde{X};\theta_0 + \Delta \theta) - f_{jk}(\tilde{X};\theta_0)) \right|\right)}{\exp(t/4\phi)}.     
 \end{align*}
 where $\phi \equiv t/(16 J^2 G B_{2n}^2 \sup_{\theta}\|\theta\|^2_1)$.
 
 By the mean value theorem, there exists a mean value vector $\tilde{\theta}$ as a function of $\Delta \theta$ and its corresponding index vector $\tilde{\nu}$ as a function of $\Delta \nu \equiv \nu - \nu_0$ given the fixed matrix of $X$ such that
 \[
 f_{jk}(\tilde{X};\theta_0 + \Delta \theta) - f_{jk}(\tilde{X};\theta_0) = \sum_{g=0}^{G} \sum_{j' \in [J]} \frac{d\xi_{j}(\tilde{\nu};\tilde{X})}{d\nu_{j'g}}(\nu_{j'g} - \nu_{j'g0})h_{jk}(W).
 \]
 Let $\mathcal{N}$ be the support of $\nu$ given the conditioning $\tilde{X}$ and the parameter space $\Theta$. Therefore, we have
  \begin{align*}
   &\frac{E_{\sigma}\exp\left(\phi \max_{j \in [J],k \in [K]} \sup_{\Delta \theta \in \Theta} \left|\mathbb{G}_n \sigma (f_{jk}(\tilde{X};\theta_0 + \Delta \theta) - f_{jk}(\tilde{X};\theta_0)) \right|\right)}{\exp(t/4\phi)}\\
   &\leq 
   \frac{E_{\sigma}\exp\left(\phi \max_{j \in [J],k \in [K]} \sup_{\Delta \nu \in \mathcal{N}} \left|\mathbb{G}_n \sigma \sum_{g=0}^{G} \sum_{j' \in [J]}
    \frac{d\xi_{j}(\tilde{\nu};\tilde{X})}{d\nu_{j'g}}(\nu_{j'g} - \nu_{j'g0})h_{jk}(W) \right|\right)}{\exp(t/4\phi)}      \\
   &\leq 
   \frac{E_{\sigma}\exp\left(\phi \sum_{g=0}^{G} \sum_{j' \in [J]} \max_{j,j' \in [J],g \in [G],k \in [K]} \sup_{\Delta \nu \in \mathcal{N}} \left|\mathbb{G}_n \sigma
    \frac{d\xi_{j}(\tilde{\nu};\tilde{X})}{d\nu_{j'g}}(\nu_{j'g} - \nu_{j'g0})h_{jk}(W) \right|\right)}{\exp(t/4\phi)}      \\    
   &\leq 
   \frac{E_{\sigma}\exp\left(JG \phi \max_{j,j' \in [J],g \in [G],k \in [K]} \sup_{\Delta \nu \in \mathcal{N}} \left|\mathbb{G}_n \sigma 
    \frac{d\xi_{j}(\tilde{\nu};\tilde{X})}{d\nu_{j'g}}(\nu_{j'g} - \nu_{j'g0})h_{jk}(W) \right|\right)}{\exp(t/4\phi)}      \\        
   &\leq 
   \frac{J^2 GK \max_{j,j' \in [J],g \in [G],k \in [K]} E_{\sigma}\exp\left(JG \phi \sup_{\Delta \nu \in \mathcal{N}} \left|\mathbb{G}_n \sigma
    \frac{d\xi_{j}(\tilde{\nu};\tilde{X})}{d\nu_{j'g}}(\nu_{j'g} - \nu_{j'g0})h_{jk}(W) \right|\right)}{\exp(t/4\phi)} .
  \end{align*}
  The first inequality follows from the mean value theorem and the fact that the supremum over $\Delta \theta$ is dominated by the supremum over $\Delta \nu$ whcih are constrained conditional on each realization of the matrix $X$. The second inequality follows from the triangular inequality. The third inequality follows from the union bound over $j',g$ by taking the maximum over $j',g$ indices. Finally, we take the union bound over $j,j',g,k$ indices.

 To apply Ledoux-Talagrand contraction inequality (Theorem 4.12, \citealp{ledouxProbabilityBanachSpaces1991}) in bounding the following term
 \begin{align*}
     &\exp\left(JG \phi \sup_{\Delta \nu \in \mathcal{N}} \left|\sum_{i \in [n]} \sigma_i  \frac{d\xi_{ij}(\tilde{\nu}_i;\tilde{X}_i)}{d\nu_{j'g}}(\Delta \nu_{i,j'g})h_{jk}(W_i) \right|\right),
 \end{align*}
 let 
 \[
  \psi_i^{jkj'g}(\Delta \nu_{ij'g}) \equiv \left[\frac{d\xi_j(\tilde{\nu}_{i};\tilde{X}_i)\Delta \nu_{j'g}}{d \nu_{j'g}}\right]h_{jk}(W_i)
 \]
 then by lemma \ref{lemma.Lip}, we have
 \[
  |\psi_i^{jkj'g}(\Delta \nu_{ij'g})| \leq C_1 J |\Delta \nu_{ij'g}|
 \]
 uniformly over $\tilde{\nu_{i}}$.
 
 Since $\psi_i^{jkj'g}(0) = 0$, Ledoux-Talagrand contraction inequality via corollary \ref{crr.LT} applies so that
 \begin{align*}
   &\frac{J^2 G K \max_{j',j \in [J],k \in [K],g \in [G]} E_{\sigma}\exp\left(JG \phi  \sup_{\Delta \nu_{j'g}\in \mathcal{N}}  \left|\mathbb{G}_n \sigma  \left(\frac{d\xi_j(\tilde{\nu};\tilde{X})}{d \nu_{j'g}}\right)\Delta \nu_{j'g}h_{jk}(W) \right|\right)}{\exp(t/4\phi)}\\
   &\leq \frac{J^2 G K \max_{j',j \in [J],k \in [K],g \in [G]}E_{\sigma}\exp(C_1 J^2 G \phi \sup_{\Delta \nu_{j'g} \in \mathcal{N}} \left|\mathbb{G}_n \sigma  \Delta \nu_{j'g} h_{jk}(W) \right|)}{\exp(t/4\phi)}\\
    &\leq \frac{J^2 G K \max_{j',j \in [J], k \in [K], g \in [G]} E_{\sigma}\exp(C_1 J^2 G \phi \sup_{\Delta \theta_g \in \Theta_g} \left|\mathbb{G}_n \sigma  h_{jk}(W) X_{j',g} \Delta \theta_{g}  \right|)}{\exp(t/4\phi)}.
 \end{align*}

 By Holder inequality, we have
\begin{align*} 
    &\frac{J^2 GK  \max_{j',j \in [J], k \in [K],g \in [G]}E_{\sigma} \exp(C_1 J^2 G \phi  \sup_{\Delta \theta_g \in \Theta_g} \left|\mathbb{G}_n \sigma  h_{jk}(W) X_{j',g}\Delta \theta_{g} \right|)}{\exp(t/4\phi)}\\
    &\leq \frac{J^2 K G \max_{j',j \in [J],k \in [K]}E_{\sigma}\exp(C_1 J^2 G \phi \max_{l \in [L]} \left|\mathbb{G}_n \sigma h_{jk}(W) X_{j',l}\right| \sup_{\Delta \theta}\|\Delta \theta \|_1 )}{\exp(t/4\phi)}\\
    &\leq J^2 G K L \max_{j,j' \in [J], k \in [K],l \in [L]}\frac{E_{\sigma}\exp(C_1 J^2 G \phi \left|\mathbb{G}_n \sigma h_{jk}(W) X_{j',l}\right| \sup_{\Delta \theta}\|\Delta \theta\|_1 )}{\exp(t/4\phi)}\\            
    &\leq J^2 G K L \max_{j,j'\in [J],k \in [K],l \in [L]}\frac{2\exp(2C_1^2 J^4 G^2 \phi^2 E_n (h_{jk}(W)^2 X_{j',l}^2) \sup_{\Delta \theta \in \Theta}\|\Delta \theta\|_1^2 )}{\exp(t/4\phi)}\\                
    &\leq J^2 G K L \max_{j,j' \in [J],k \in [K],l \in [L]}\frac{2\exp(2C_1^2 J^4 G^2 \phi^2 B_{2n}^2 \sup_{\theta \in \Theta}\|\theta\|_1^2 )}{\exp(t/4\phi)}
 \end{align*}
 from the symmetry of distribution and sub-Gaussianity. Then the stated bound is achieved by following the analogue argument of Lemma D.3 of \cite{belloniHighdimensionalEconometricsRegularized2018}.
\end{proof}

Then the following statement shows the error rate of RGMM BLP estimator:
\begin{theorem}
 Under assumption \ref{ass.nlid}, and assumptions for theorem \ref{thm.tailbound}. Assume further that
 \[
  \|\hat{f}(\theta_0)\|_{\infty} \leq \lambda
 \]
  with probability at least $1 - \alpha$, then we have for each $q \in \{1,2\}$
 \[
  \|\theta - \theta_0\|_q \leq 2 s^{1/q} \mu_n^{-1} n^{-1/2} (\tilde{l}_n + l_n)
 \]
 with probability at least $1 - \alpha - \delta_n$, where
 \[
  \tilde{l}_n \equiv C(B_{1n} + (J^2G)(2\sqrt{2}B_{2n}\sup_{\theta}\|\theta\|_1 \log^{1/2}(8J^2 GKL/\delta_n)))
 \]
 with a universal constant $C$.
\end{theorem}

\section{De-biased RGMM}

Given the RGMM estimator $\hat{\theta}$, it is recommended that we update the estimate in order to make a proper inference. De-biased Lasso, or De-biased RGMM procedure in \cite{belloniHighdimensionalEconometricsRegularized2018} takes the following steps
\begin{enumerate}
    \item Estimate the RGMM $\hat{\theta}$
    \item Estimate the plug-in gradient
          \[
           \hat{G} = \partial_{\theta'} \hat{f}(\hat{\theta})
          \]
          and the plug-in var-cov matrix
          \[
           \hat{\Omega} = \mathbb{E}_n f(\tilde{X};\hat{\theta})f(\tilde{X};\hat{\theta})'
          \]
    \item Solve the minimization problem of
        \[
         \min_{\gamma \in \mathbb{R}^{2L \times JK}} \sum_{l \in [2L]}\|\gamma_l \|_1
        \]
        subject to
        \[
         \|\gamma_l \hat{\Omega} - (\hat{G}')_l \|_{\infty} \leq \lambda_l^{\gamma}
        \]
        for some regularization parameters $\lambda_l^{\gamma}$
    \item Solve the minimization problem of
        \[
         \min_{\mu \in \mathbb{R}^{JK \times 2L}} \sum_{j \in [JK]}\|\mu_j \|_1
        \]
        subject to
        \[
         \|\mu_j \hat{\gamma} \hat{G} - e'_j \|_{\infty} \leq \lambda_j^{\mu}
        \]
        for some regularization parameters $\lambda_j^{\mu}$, where $e_j$ is a coordinate vector with 1 in the $j$-th position and 0 elsewhere.
   \item Update the RGMM estimator as $\hat{\theta} - \hat{\mu}\hat{\gamma} \hat{f}(\hat{\theta})$.
\end{enumerate}

First of all, we need to provide maximal inequalities for the auxiliary estimators $\hat{\gamma}$ and $\hat{\mu}$. The strategy follows the parallel argument of the maximal inequality for $\hat{\theta}$. Therefore, we need the following modulus of continuity conditions for $\gamma$ and $\mu$.
\begin{assumption}\label{ass.nlid2}
 Suppose that there exists a sequence $\mu_n$ such that
 \[
  \inf_{\gamma \in \mathcal{R}(\gamma_0):\|\gamma - \gamma_0\|_1 > 0}\|(\gamma - \gamma_0)\Omega\|_{\infty}/\|\gamma - \gamma_0\|_1 \geq s^{-1}\mu_n
 \]
 and
 \[
  \inf_{\mu \in \mathcal{R}(\mu_0):\|\mu - \mu_0\|_1 > 0}\|(\mu - \mu_0)G'\Omega G\|_{\infty}/\|\mu - \mu_0\|_1 \geq s^{-1}\mu_n.
 \] 
 \end{assumption}
Note that the first condition requires that the variance matrices constructed from any elements of the score functions $f_{jk}(X,\theta_0)$ is non-singular, and the second condition follows if all eigenvalues of $G'\Omega G$ are bounded in absolute values from zero uniformly over $n$.

Then given a choice of penalty parameters $\lambda_l^{\gamma}$ and $\lambda_j^{\mu}$, we achieve the maximal inequality for $L^1$-norm of $\hat{\gamma}_l$ and $\hat{\mu}_j$
\begin{lemma}[Lemma 3.7 of \cite{belloniHighdimensionalEconometricsRegularized2018}]
 Let $l_{n}^{\Omega}$ and $l_{n}^G$ such that
 \[
  n^{1/2} \|\hat{\Omega}  - \Omega\|_{\infty} \leq l_{n}^{\Omega}
 \]
 and
 \[
  n^{1/2} \|\hat{G}  - G\|_{\infty} \leq l_{n}^{G}
 \]
 with probability $1 - \delta_n$. Suppose $\max_{l \in [2L]} \|\gamma_{0l}\|_1 \leq \bar{C}$, and $\max_{j \in [JK]} \|\mu_{0j}\|_1 \leq \bar{C}$.
 
 Let $\lambda_l^{\gamma}$ satisfy
 \begin{align*}
  n^{1/2}\lambda_l^{\gamma} &\geq \bar{C} l_n^{\Omega} + l_n^{G},\\
  \lambda_l^{\gamma} &\leq n^{-1/2}l_n
  \end{align*}
  and $\lambda_j^{\mu}$ satisfy
 \begin{align*}  
  n^{1/2}\lambda_j^{\mu} &\geq 2\bar{C}^2 l_n^G + \bar{C}^3 l_n^{\Omega} + \bar{C}^2 \max_{l \in [2L]} n^{1/2}\lambda_l^{\gamma}\\
  \lambda_j^{\mu} &\leq n^{-1/2}l'_n.
 \end{align*}
 for $l \in [2L]$ and $j \in [JK]$. Suppose assumption \ref{ass.nlid2} holds. Then with probability $1 - 3\delta_n$, we have
 \[
  \max_{l \in [2L]} \|\hat{\gamma}_l - \gamma_{0l}\|_1 \leq \frac{ s l_n (2 + \bar{C})}{\mu \sqrt{n}}
 \]
 and with probability $1 - \delta_n$
 \[
  \max_{j \in [JK]} \|\hat{\mu}_j - \mu_{0j}\|_1 \leq \frac{s l'_n (2 + \bar{C})}{\mu_n \sqrt{n}}.
 \]
 
\end{lemma}

Next, we need maximal inequalities for the norms $\|\hat{G} - G\|_{\infty}, \|\hat{G} - \tilde{G}\|_{\infty},$ and $\|\hat{\Omega} - \Omega\|_{\infty}$ where $\tilde{G} \equiv -\hat{G}(\tilde{\theta})$ with $\tilde{\theta}$ as the intermediate value of $\hat{\theta}$ and $\theta_0$. Unlike lemma 3.7 of \cite{belloniHighdimensionalEconometricsRegularized2018} which assumes the tail probability bound for the process $\hat{\Omega} - \Omega$ and $\hat{G} - G$, we need certain modification of lemma as we do for theorem \ref{thm.tailbound}.

\begin{lemma}
 Suppose the following 
 \begin{enumerate}
  \item 
  \[
   \max_{j \in [J],k \in [K],l \in [L]} E\left[h_{jk}^2(W) \max_{j'' \in [J]} X_{j'',l}^2 \max_{j' \in [J], l' \in [L]}X_{j',l'}^2\right] \leq C,
  \]
  \item  with probability $1 - \delta_n$, we have
  \[
   \max_{j \in [J],k \in [K],l \in [L]} \mathbb{E}_n\left[h_{jk}^2(W) \max_{j'' \in [J]} X_{j'',l}^2 \max_{j' \in [J], l' \in [L]}X_{j',l'}^2\right] \leq B_n^2,
  \]
  \[
   \max_{j \in [J],k \in [K],l \in [L]} \mathbb{E}_n[h_{jk}^2(W) X_{jl}^2] \leq B_{n}^2
  \]
  and
  \[
   \max_{j,j' \in [J],k,k' \in [K],l \in [L]} \mathbb{E}_n[h_{jk}^2(W)h_{j',k'}^2(W) X_{jl}^2 \xi_{j'}(\tilde{X};\theta_0)^2] \leq B_{n}^2
  \]
  with probability $1 - \delta_n$.
  \item with probability $1 - \delta_n$, we have
    \[
     \|\hat{\theta} - \theta_0\|_q \leq \Delta_{qn}
    \]
    for $q \in \{1,2\}$.
  \item $\max_{j \in [J], k \in [K]} E[f_{jk}^4(\tilde{X};\theta_0)] \leq C$, and
        $n^{-1/2} E\left[\max_{i \in [n]} \|f(\tilde{X}_i;\theta_0)\|_{\infty}^4\right] \leq \min\{\delta_n,\log^{-1/2}(JK)\}$
  \item \[
            \max_{j \in [J], k \in [K],l\in [2L]}E\left[G_{jk,l}(\tilde{X};\theta_0)^2\right] \leq C
        \]
        and
        \[
         n^{-1/2}E[\max_{i \in [n]}\|G(\tilde{X}_i;\theta_0)\|^2_{\infty}] \leq \min\{\delta_n, \log^{-1/2}(2JKL)\}.
        \]
 \end{enumerate} 
 Then with probability $1 - C'\delta_n$ we have
 \[
  \|\hat{G} - G\|_{\infty} \leq C'\sqrt{n^{-1} \log(2JKL)} + C'J^2 G B_n \Delta_{1n} \sqrt{n^{-1}\log(2J^2 GKL/\delta_n)} + C'J^{3/2}\Delta_{2n}.
 \]
  \[
  \|\hat{G} - \tilde{G}\|_{\infty} \leq C'\sqrt{n^{-1} \log(2JKL)} + C'J^2 G B_n \Delta_{1n} \sqrt{n^{-1}\log(2J^2 GKL/\delta_n)} + C'J^{3/2}\Delta_{2n}
 \]
 and
 \[
  \|\hat{\Omega} - \Omega\|_{\infty} \leq C'\sqrt{n^{-1}\log(JK)} + C'J^2 G B_n \Delta_{1n}\sqrt{n^{-1}\log(2J^2 GKL/\delta_n)} + 2J^{3/2} C (J^{3/2}\Delta_{2n}^2 + \Delta_{2n}).
 \]
\end{lemma}
\begin{proof}
 We follow the proof of lemma 3.9 of Belloni et al. (2018). First we bound $\|\hat{G} - G\|_{\infty}$. Note that $\sqrt{n}\|\hat{G} - G\|_{\infty}$ is bounded by the sum of the following three terms from triangular inequality:
 \begin{enumerate}
     \item[(1.1)] $\max_{j \in [J],k \in [K],l \in [2L]}|\mathbb{G}_n (G_{jk,l}(\tilde{X};\hat{\theta}) - G_{jk,l}(\tilde{X};\theta_0))|$
     \item[(1.2)] $\max_{j \in [J],k \in [K],l \in [2L]}|\mathbb{G}_n G_{jk,l}(\tilde{X};\theta_0))|$
     \item[(1.3)] $\max_{j \in [J],k \in [K],l \in [2L]} n^{1/2}|E[G_{jk,l}(\tilde{X};\hat{\theta}) - G_{jk,l}(\tilde{X};\theta_0)]|$
 \end{enumerate} 
 
 From lemma C.1 (4) of Belloni et al. (2018), the second term, (1.2), is bounded by
 \begin{align*}
  \max_{j \in [J],k \in [K],l \in [2L]}|\mathbb{G}_n G_{jk,l}(\tilde{X};\theta_0))|  \leq & C \max_{j,k,l} \frac{1}{n}\sum_i E\left[|G_{jk,l}(\tilde{X}_i;\theta_0))|^2\right]\log^{1/2}(2JKL) \\
  &+ n^{-1/2}C_2 \left\{E\left[\max_{i\leq n}\|G(\tilde{X}_i;\theta_0))\|^2_{\infty}\right](\delta_n^{-1} + \delta_n^{-1/2} + log(2JKL)) \right\}\\
  \leq& C \log^{1/2}(2JKL)
 \end{align*}
 with probability $1 - \delta_n$ by the condition 5.
 
 For the last term, (1.3), we use the linear expansion of $G_{jk,l}(\tilde{X};\theta)$ into $\sum_{j' \in [J],g \in [G]} h_{jk}(W)B_{j',g}(X_l)(\nu_{j'g} - \nu_{j'g0})$ from lemma \ref{lemma.LipGrad} so that
 \begin{align*}
    &\max_{j \in [J],k \in [K],l \in [2L]} n^{1/2}\left|E[G_{jk,l}(\tilde{X},\hat{\theta}) - G_{jk,l}(\tilde{X},\theta_0)]\right|\\
    &\leq_{(1)} \max_{j \in [J],k \in [K],l \in [2L]} n^{1/2}\left|E\left[\sum_{j' \in [J],g \in [G]} h_{jk}(W)B_{j',g}(X_l)(\nu_{j'g} - \nu_{j'g0})\right]\right|\\
    &\leq_{(2)} \max_{j \in [J],k \in [K],l \in [2L]} n^{1/2}\left|E\left[\sum_{j' \in [J],g \in [G]} h_{jk}(W)B_{j',g}(X_l)X_{j'g}'(\theta_{g} - \theta_{g0})\right]\right|\\    
    &\leq_{(3)} \max_{j \in [J],k \in [K],l \in [2L]} n^{1/2}\left|E\left[\sum_{j' \in [J],l' \in [L]} [h_{jk}(W)B_{j',g(l')}(X_l) X_{j',l'} (\theta_{l'} - \theta_{l'0})]^2\right]\right|^{1/2}\\
&\leq_{(4)} \max_{j \in [J],k \in [K],l \in [2L]} n^{1/2}\left|E\left[h^2_{jk}(W) \sum_{j' \in [J],l' \in [L]} \max_{j' \in [J],l' \in [L]}[B^2_{j',g(l')}(X_l) X^2_{j',l'}] (\theta_{l'} - \theta_{l'0})^2\right]\right|^{1/2}\\
&\leq_{(5)} \max_{j \in [J],k \in [K],l \in [2L]} n^{1/2}\left|E\left[h^2_{jk}(W) J^2 \bar{C}^2 \max_{j'' \in [J]} X_{j'',l}^2 \max_{j' \in [J],l' \in [L]} X^2_{j',l'} \sum_{j' \in [J],l' \in [L]} |\theta_{l'} - \theta_{l'0}|^2\right]\right|^{1/2}\\
&\leq_{(6)} \max_{j \in [J],k \in [K],l \in [2L]} n^{1/2}\bar{C} J^{3/2} \sqrt{E\left[h^2_{jk}(W) \max_{j'' \in [J]} X_{j'',l}^2 \max_{j' \in [J],l' \in [L]} X^2_{j',l'}\right]} \|\theta - \theta_{0}\|_2 \leq C n^{1/2} J^{3/2} \Delta_{2n}
\end{align*}
with probability $1 - \delta_n$, where (1) and (5) follows from lemma \ref{lemma.LipGrad}, (2) follows from the definition of $\nu_{jm}$, (3) follows from the monotonicity of the $L_p$ norm, (4) follows from the union bound, and (6) follows from the condition 1.
 
 The first term, (1.1), is the empirical process in terms of the $G_{jk,l}$ instead of $f_{jk}$. Note that
 \begin{align*}
  \mathbb{E}_n& Var(G_{jk,l}(\tilde{X};\hat{\theta}) - G_{jk,l}(\tilde{X};\theta_0)) \leq \mathbb{E}_n \left(\sum_{j' \in [J],l \in [L]}h_{jk}(W)B_{j',g(l)}X_{j',l}(\theta_l - \theta_{l0})\right)^2\\
  \leq & \mathbb{E}_n \left[ h_{jk}^2(W) \max_{j' \in [J],l \in [L]} B_{j',g(l)}^2 X^2_{j',l}\right] \|\theta - \theta_0\|_1^2\\
  \leq & C J^2 \mathbb{E}_n \left[h_{jk}^2(W) \max_{j'' \in [J]}  X^2_{j'',l} \max_{j' \in [J],l \in [L]} X^2_{j',l}\right] \|\theta - \theta_0\|_1^2 \leq C J^2 B^2_n \Delta_{1n}^2
 \end{align*}
 by H\"older inequality and lemma \ref{lemma.LipGrad}.
 
 Then the conditions for the Corollary \ref{crr.contraction2} holds with $B_{1n} = J B_n \Delta_{1n}$, and $B_{2n} = B_n$ so that
 \[
  \max_{j \in [J],k \in [K],l \in [2L]}|\mathbb{G}_n (G_{jk,l}(\tilde{X};\hat{\theta}) - G_{jk,l}(\tilde{X};\theta_0))| \leq J^2 G  C B_n\Delta_{1n} \log^{1/2}(2J^2 G KL/\delta_n)
 \]
 with probability $1 - 7\delta_n$.
 
 In the same argument of Belloni et al (2018), $\|\hat{G} - \tilde{G}\|_{\infty}$ has the same bound as $\|\hat{G} - G\|_{\infty}$. Finally, we consider $\|\hat{\Omega} - \Omega\|$. As we do for $\|\hat{G} - \tilde{G}\|_{\infty}$, $n^{-1/2}\|\hat{\Omega} - \Omega\|$ is bounded by the sum of the following three terms
 \begin{enumerate}
     \item[(2.1)] $\max_{j,j' \in [J],k,k' \in [K]}|\mathbb{G}_n (f_{jk}(\tilde{X};\hat{\theta})f_{j'k'}(\tilde{X};\hat{\theta}) - f_{jk}(\tilde{X};\theta_0)f_{j'k'}(\tilde{X};\theta_0)|$
     \item[(2.2)] $\max_{j,j' \in [J],k,k' \in [K]}|\mathbb{G}_n f_{jk}(\tilde{X};\theta_0)f_{j'k'}(\tilde{X};\theta_0)|$
     \item[(2.3)] $\max_{j,j' \in [J],k,k' \in [K]} n^{1/2}|E[f_{jk}(\tilde{X};\hat{\theta})f_{j'k'}(\tilde{X};\hat{\theta}) - f_{jk}(\tilde{X};\theta_0)f_{j'k'}(\tilde{X};\theta_0)]|$.
 \end{enumerate} 
 
 The first term, (2.1), is further bounded by the sum of two terms
 \begin{enumerate}
     \item[(2.1.1)] \[
            \max_{j,j' \in [J],k,k' \in [K]}|\mathbb{G}_n (f_{jk}(\tilde{X};\hat{\theta}) - f_{jk}(\tilde{X};\theta_0))(f_{j'k'}(\tilde{X};\hat{\theta}) - f_{j'k'}(\tilde{X};\theta_0))|     
           \]
     \item[(2.1.2)] \[
            2\max_{j,j' \in [J],k,k' \in [K]}|\mathbb{G}_n ((f_{jk}(\tilde{X};\hat{\theta}) -  f_{jk}(\tilde{X};\theta_0))f_{j'k'}(\tilde{X};\theta_0)|.
           \]
 \end{enumerate}  
 First, we have for (2.1.1),
 \begin{align*}
     &\max_{j,j' \in [J],k,k' \in [K]}|\mathbb{G}_n (f_{jk}(\tilde{X};\hat{\theta}) - f_{jk}(\tilde{X};\theta_0))(f_{j'k'}(\tilde{X};\hat{\theta}) - f_{j'k'}(\tilde{X};\theta_0))| \\
     &\leq  \max_{j,j' \in [J],k,k' \in [K]}n^{1/2}\mathbb{E}_n[(f_{jk}(\tilde{X};\hat{\theta}) - f_{jk}(\tilde{X};\theta_0))^2]^{1/2} \mathbb{E}_n[(f_{j'k'}(\tilde{X};\hat{\theta}) - f_{j'k'}(\tilde{X};\theta_0))^2]^{1/2} \\ 
     & + \max_{j,j' \in [J],k,k' \in [K]}n^{1/2}E[(f_{jk}(\tilde{X};\hat{\theta}) - f_{jk}(\tilde{X};\theta_0))^2]^{1/2} E[(f_{j'k'}(\tilde{X};\hat{\theta}) - f_{j'k'}(\tilde{X};\theta_0))^2]^{1/2} \\ 
     &\leq  \max_{j,j' \in [J],k,k' \in [K]}n^{1/2}C_1^2 J^2\mathbb{E}_n\left[ h_{jk}^2(W)\left(\sum_{j'' \in [J],l \in [L]}X_{j'',l}(\hat{\theta}_l - \theta_{l0})\right)^2\right]^{1/2}\mathbb{E}_n\left[ h_{j'k'}^2(W) \left(\sum_{j''\in [J],l \in [L]}X_{j'',l}(\hat{\theta}_l - \theta_{l0})\right)^2\right]^{1/2} \\ 
     & + \max_{j,j' \in [J],k,k' \in [K]}n^{1/2}C_1^2  J^2E\left[ h_{jk}^2(W) \left(\sum_{j'' \in [J],l \in [L]}X_{j'',l}(\hat{\theta}_l - \theta_{l0})\right)^2\right]^{1/2}E\left[ h_{j'k'}^2(W) \left(\sum_{j'' \in [J],l \in [L]}X_{j'',l}(\hat{\theta}_l - \theta_{l0})\right)^2\right]^{1/2} \\ 
     & \leq n^{1/2} C B_n^2 J^{3} \Delta_{2n}^2
 \end{align*}
 with probability $1 - C\delta_n$ and a universal constant $\bar{C}$, where the second and the last inequalities are by H\"older inequality and lemma \ref{lemma.Lip}.
 
 For the second term, (2.1.2), note that
 \begin{align*}
  &Var(\mathbb{G}_n([f_{j,k}(\tilde{X};\theta) - f_{j,k}(\tilde{X};\theta_0)]f_{j',k'}(X,\theta_0))) \leq Var(\mathbb{G}_n(h_{jk}(W)h_{j'k'}(W)\xi_{j'}(\tilde{X};\theta_0)X_{j}'(\theta - \theta_0))) \\
  &\leq \bar{C} J^2 B_n^2 \Delta_{1n}^2
 \end{align*}
 so that the conditions for corollary \ref{crr.contraction} holds with $B_{1n} = JB_n \Delta_{1n}$ and $B_{2n} = B_n$, therefore, 
 \[
  \max_{j,j' \in [J],k,k' \in [K]}|\mathbb{G}_n ((f_{jk}(\tilde{X};\hat{\theta}) - f_{jk}(\tilde{X};\theta_0))f_{j'k'}(\tilde{X};\theta_0)| \leq \bar{C}_1 J^2 G B_n \Delta_{1n} \log^{1/2}(2J^2GKL/\delta_n).
 \]
 
 For the remaining two terms, (2.2),
 \[
 \max_{j,j' \in [J],k,k' \in [K]}|\mathbb{G}_n f_{jk}(\tilde{X};\theta_0)f_{j'k'}(\tilde{X};\theta_0)|
 \]
 is bounded by 
 \[
  C\max_{j,k} E[f_{j,k}^4(\tilde{X};\theta_0)]^{1/2} \sqrt{\log(JK)} + Cn^{-1/2}E[\max_{i}\|f(\tilde{X}_i;\theta_0)\|^4_{\infty}]\{\delta_n^{-1} + \log(JK)\} \leq C' \sqrt{\log(JK)}
 \]
 from Lemma C.1(4) of Belloni et al. (2018) under condition 4.
 
 Finally, for (2.3),
 \begin{align*}
 &\max_{j,j' \in [J],k,k' \in [K]} n^{1/2}|E[f_{jk}(\tilde{X};\hat{\theta})f_{j'k'}(\tilde{X};\hat{\theta}) - f_{jk}(\tilde{X};\theta_0)f_{j'k'}(\tilde{X};\theta_0)]|\\
 \leq & \max_{j,j' \in [J],k,k' \in [K]} n^{1/2}|E[f_{jk}(\tilde{X};\hat{\theta})(f_{j'k'}(\tilde{X};\hat{\theta}) - f_{j'k'}(\tilde{X};\theta_0))]| + |E[(f_{jk}(\tilde{X};\hat{\theta}) - f_{jk}(\tilde{X};\theta_0))f_{j'k'}(\tilde{X};\theta_0)]|\\
\leq & \max_{j,j' \in [J],k,k' \in [K]} n^{1/2}|E[(f_{jk}(\tilde{X};\hat{\theta}) - f_{jk}(\tilde{X};\theta_0))(f_{j'k'}(\tilde{X};\hat{\theta}) - f_{j'k'}(\tilde{X};\theta_0))]| \\
&+ 2 n^{1/2} \max_{j,j' \in [J],k,k' \in [K]} |E[(f_{jk}(\tilde{X};\hat{\theta}) - f_{jk}(\tilde{X};\theta_0))f_{j'k'}(\tilde{X};\theta_0)]|\\ 
&\leq J^{3} C \Delta_{2n}^2 + n^{1/2} 2 \max_{j,j' \in [J],k,k' \in [K]} E\left[\bar{C}^2J^2 h^2_{jk}(W)h^2_{j'k'}(W)\xi^2_{j'}(\tilde{X};\theta_0)\left(\sum_{j' \in [J],g \in [G]} X'_{j',g}(\theta_g - \theta_{g0})\right)^2\right]^{1/2}\\
&\leq n^{1/2} J^{3} C \Delta_{2n}^2 + n^{1/2} 2 \max_{j,j' \in [J],k,k' \in [K]} E\left[\bar{C}^2J^2 h^2_{jk}(W)h^2_{j'k'}(W)\xi^2_{j'}(\tilde{X};\theta_0) \max_{j' \in [J],l' \in [L]}X_{j',l'}^2 \sum_{j' \in [J],l' \in [L]}(\theta_l - \theta_{l0})^2\right]^{1/2}\\
&\leq n^{1/2} J^{3} C \Delta_{2n}^2 + n^{1/2} 2 J^{3/2} C \Delta_{2n} \\
 \end{align*}
 by lemma \ref{lemma.Lip} and H\"older inequality.
\end{proof}

Combining these results, we attain the asymptotic linearlity
\begin{theorem}
 Suppose that
 \[
 \max_{j \in [JK]}E[f^2_{j}(\tilde{X};\theta_0)] \leq C
 \]
 and
 \[
 n^{-1/2}E\left[\max_{i \in [n]}\|f(\tilde{X}_i;\theta_0)\|^2_{\infty}\right] \leq \min\{\delta_n, \log^{-1/2}(JKL)\}.
 \]
 
 Suppose that $\|\hat{f}(\theta_0)\|_{\infty} \leq \lambda$ with probability at least $1 - \alpha$, and assumptions for lemmas 1, 2 and 3, and theorem \ref{thm.tailbound} with
 \[
   B_n + \bar{C} + \mu_n^{-1} \leq C.
 \]
 For $\bar{a} \geq 0$ and $C' \geq 1$, let 
 \[
  \bar{\lambda} = C'J^{3/2}\max\{J^{3/2}\tilde{\lambda}^2,\tilde{\lambda}\}
 \]
 with $\tilde{\lambda} \equiv n^{-1/2+\bar{a}}J^2 G \Phi^{-1}(1 - (2J^2GKLn)^{-1})$.
 
 Then, setting $\lambda_l^{\gamma} = \frac{1}{2}\lambda_{l}^{\mu} = \bar{\lambda}$, we have with probability $1 - \alpha - C\delta_n$,
 \[
  \sqrt{n}(\hat{\hat{\theta}} - \theta_0) = -\mu_0 \gamma_0 \hat{f}(\theta_0) + r
 \]
 with $\|r\|_{\infty} \leq C u_n$ where $\hat{\hat{\theta}}$ is the updated RGMM estimator provided that
 \[
  n^{-1+2\bar{a}} s J^{3/2} \max\left\{n^{-1/2}J^{3/2}J^4 G^2 \log^2(2J^2GKLn),J^2 G \log(2J^2GKLn) \right\} \leq u_n
 \]
 and
 \[
  \bar{\lambda} \geq C J^{3/2} \max\{J^{3/2}D_{2n}^2,D_{2n}\}
 \]
 for some large enough $C > 0$ where $D_{2n} \equiv n^{-1/2} s^{1/2}(\log^{1/2}(2JKL) + J^2 G \log^{1/2}(2J^2GKL/\delta_n))$.
\end{theorem}
\begin{proof}
 First note that the rate terms of theorem \ref{thm.tailbound} satisfies the following
 \[
  l_n \leq  C'\log^{1/2}(2JKL)
 \]
 and
 \[
  l'_n \leq C' J^2 G \log^{1/2}(2J^2 K G L) 
 \]
 with probability $1 - \delta_n$ under the assumptions.
 
 Therefore,
 \[
  \Delta_{qn} \leq C' n^{-1/2} s^{1/q} (l_n + l'_n) \leq C'' n^{-1/2} s^{1/q}  (\log^{1/2}(2JKL) + J^2 G\log^{1/2}(2J^2 K G L)) \equiv D_{qn}.
 \]
 with probability at least $1 - \alpha - \delta_n$.
 
 To apply lemma 2, note that for $s \geq 1$, we have
 \[
  n^{-1/2} \log^{1/2}(2JKL) \leq n^{-1/2} s^{1/q} \log(2JKL) \leq D_{qn}.
 \]
 Note also that
 \[
  D_{1n}n^{-1/2}J^2 G \log^{1/2}(2J^2 GKL/\delta_n) \leq D_{2n}^2.
 \]
 Therefore, by lemma 3, we have
 \[
  \max\{\|\hat{G} - G\|_{\infty}, \|\hat{G} - \tilde{G}\|_{\infty}, \|\hat{\Omega} - \Omega\|_{\infty}\} \leq C J^{3/2}(J^{3/2}D_{2n}^2 + D_{2n}).
 \]
 with probability $1 - C\delta_n$.
 
 Now, let
 \[
  l_n^{\Omega} = l_n^{G} = n^{1/2} C J^{3/2} \max\{J^{3/2}D_{2n}^2,D_{2n}\}
 \]
 and
 \[
  \lambda_{l}^{\gamma} = \frac{1}{2}\lambda_{j}^{\mu} = \bar{\lambda}
 \]
 so that
 \[
  n^{1/2}\bar{\lambda} \geq (\bar{C} + 1)n^{1/2} C J^{3/2}(J^{3/2}D_{2n}^2 + D_{2n})
 \]
 and
 \[
  \bar{\lambda} \leq n^{-1/2}l_n = \bar{\lambda}.
 \]
 
 Then lemma 2 applies to get
 \[
  \max_{l \in [p]}\|\hat{\gamma}_{l} - \gamma_{0l}\|_{1} \leq C_1 s \bar{\lambda}
 \]
 and similar definition for $l'_n$ gives
  \[
  \max_{j \in [JK]}\|\hat{\mu}_{j} - \mu_{0j}\|_{1} \leq C_2 s \bar{\lambda}.
 \]
 
 By lemma 3.6 of Belloni et al. (2018), the decomposed error rates defined in the lemma, $\bar{r}_1, \bar{r}_2$ and $\bar{r}_3$ are bounded by
 \[
  \bar{r}_1 \leq \sqrt{n} \bar{\lambda} \Delta_{1n}
 \]
 \[
 \bar{r}_2 \leq C n^{1/2} (\Delta_{2n}^2 + J^2 \Delta_{2n}) \Delta_{1n}
 \]
 \[
  \bar{r}_3 \leq C n^{-1/2} \bar{\lambda} \Delta_{1n}.
 \]
 Thus,
 \[
  \|r\| \leq C u_n
 \]
 for $u_n$ such that $n^{-1/2 + \bar{a}} s J^{3/2} \max\{n^{-1/2}J^{3/2} J^4 G^2 \log^{3/2}(2J^2 GJKLn),J^2 G \log(2J^2 GJKLn)\} \leq u_n$ with probability $1 - \alpha - C\delta_n$.
\end{proof}

\section{Conclusion}

In this paper, we propose a $l_1$-penalized estimation for random coefficient logit model of differentiated product demands. Unlike the existing approach, our procedure allows for random coefficients on possibly high-dimensional attributes. Therefore, both of the mean and variance of the indirect utilities may be determined by high-dimensional but sparse set of attributes. 

Our strategy bases on the contraction inequality by \cite{ledouxProbabilityBanachSpaces1991} as is used in \cite{belloniHighdimensionalEconometricsRegularized2018} for the GMM procedure with a single index. We show that the $l_1$-regularized GMM estimation and its de-biased procedure are valid for the BLP model with fixed number of indices generated out of high-dimensional but sparse set of attributes. 

Unfortunately, the contraction inequality principle does not apply to a fully flexible random coefficient BLP model as the number of indices grows in the same rate as the number of the attributes. Also, our current result does not accommodate the models with the number of products growing exponentially. These challenges are left for the future work.

\bibliographystyle{ecta}
\bibliography{references} 
 







\begin{appendix}
\section{Supporting Lemmas}

We consider the bound for the linear expansion of $\xi_j$ with respect to the index $x_{j'g}'\gamma_g$ when the share of each product for any $\tilde{\beta}$ fall in the shrinking range of $c_1/J$ and $c_2/J$. This is the same assumption employed in Berry, Linton, and Pakes (2004). All these arguments should apply to the special case that the number of product $J$ is a fixed constant when every product has non-zero share in every market.
\begin{lemma} \label{lemma.Lip}
 Let
\[
 s_j(\tilde{\beta};\tilde{x},\theta) \equiv \frac{\exp(x_j'\beta + \xi_j(\tilde{x};\theta) + \sum_{g \in [G]} x_{jg}'\gamma_g \tilde{\beta}_g)}{1 + \sum_{j' \in [J]}\exp(x_{j'}'\beta + \xi_{j'}(\tilde{x};\theta) + \sum_{g \in [G]} x_{j'g}'\gamma_g \tilde{\beta}_g)}.
\]

Assume that
 \[
  \frac{c_1}{J} < s_j(\tilde{\beta};\tilde{x},\theta) < \frac{c_2}{J}
 \]
 for almost every $\tilde{\beta} \in \mathbb{R}^G$ and $x$, and for every $j \in \{0,1,\ldots,J\}$.
 
 Then for any pair of index values $\nu_{j'g}$ and $\nu_{j'g0}$ for each $g \in [G], j' \in [J]$, we have
\[
 \left|\frac{d \xi_j(\tilde{\nu};\tilde{x})}{d\nu_{j'g}}(\nu_{j'g} - \nu_{j'g0})\right| \leq 
 C_1 J \left| \nu_{j'g} - \nu_{j'g0} \right|
\]
with a universal constant $C_1$ uniformly over $\tilde{\nu}$ and $\tilde{x}$.
\end{lemma}
\begin{proof}
First note that
\[
 \frac{d\xi(\tilde{x};\theta)}{dx_{j'g}'\gamma_g} = \left[\frac{\partial s(\tilde{x};\theta)}{\partial \xi} \right]^{-1} \left[\frac{\partial s(\tilde{x};\theta)}{\partial x_{j'g}' \gamma_g} \right] = \left[\frac{\partial s(\tilde{x};\theta)}{\partial \xi} \right]^{-1}\int \tilde{\beta}_g s(\tilde{\beta};\tilde{x},\theta) (1 - s_{j'}(\tilde{\beta};\tilde{x},\theta)) dF_{\tilde{\beta}},
\]
where $s(\tilde{\beta};\tilde{x},\theta)$ is a vector of $s_j(\tilde{\beta};\tilde{x},\theta)$.

Then,
\begin{align*}
 RHS \leq& \left| \sum_{j'' \in [J]} d_{jj''}(\theta) \int\tilde{\beta}_g s_{j''}(\tilde{\beta};\tilde{x},\theta) (1 - s_{j'}(\tilde{\beta};\tilde{x},\theta)) dF_{\tilde{\beta}} (\nu_{j'g} - \nu_{j'g0}) \right| 
\end{align*}
where $d_{jj''}(\theta)$ is $(j,j'')$  element of $\left[\frac{\partial s(\tilde{x};\theta)}{\partial \xi} \right]^{-1} \equiv D(\theta)$ matrix.

For the inverse matrix elements $d_{jj''}(\theta)$, \citet[p.657]{berryLimitTheoremsEstimating2004} show that the upper bound of the $D$ matrix in the positive definite sense, i.e., 
\[
 x'\left(diag(\underline{s}_1,\ldots,\underline{s}_J)^{-1} + \frac{ii'}{\underline{s}_0} - D(\theta)\right)x > 0
\]
for any non-zero vector $x$ where $\underline{s}_j$ are the lower bounds of the shares satisfying the rate condition in the assumption. Thus, each element of $|d_{jj''}(\theta)|$ is bounded above by $\frac{1}{\underline{s_j}} + \frac{1}{\underline{s_0}} \leq 2J/c_1$.

Now,
\begin{align*}
 &\left| \sum_{j'' \in [J]} d_{jj''}(\theta) \int\tilde{\beta}_g s_{j''}(\tilde{\beta};\tilde{x},\theta) (1 - s_{j'}(\tilde{\beta};\tilde{x},\theta)) dF_{\tilde{\beta}} (\nu_{j'g} - \nu_{j'g0}) \right| \\
 & \leq \left| \sum_{j'' \in [J]} |d_{jj''}(\theta)| \int|\tilde{\beta}_g| s_{j''}(\tilde{\beta};\tilde{x},\theta) (1 - s_{j'}(\tilde{\beta};\tilde{x},\theta)) dF_{\tilde{\beta}} \right| |\nu_{j'g} - \nu_{j'g0}| \\
 & \leq  \left| \frac{2J}{c_1} \sum_{j'' \in [J]} \int|\tilde{\beta}_g| s_{j''}(\tilde{\beta};\tilde{x},\theta) dF_{\tilde{\beta}} \right| |\nu_{j'g} - \nu_{j'g0}| \\ 
& \leq  \left| \frac{2J}{c_1} \int|\tilde{\beta}_g| (1 - s_{0}(\tilde{\beta};\tilde{x},\theta)) dF_{\tilde{\beta}} \right| |\nu_{j'g} - \nu_{j'g0}| \\  
& \leq  \left| \frac{2J}{c_1} \int|\tilde{\beta}_g| dF_{\tilde{\beta}} \right| |\nu_{j'g} - \nu_{j'g0}| \leq C_1 J |\nu_{j'g} - \nu_{j'g0}|.
\end{align*}
\end{proof}
\begin{remark}
 One may achieve $L_{\infty}$-Lipschitz result for the vector of indices $\{\nu_{jg}\}_{j \in [J]}$. The $L_{\infty}$-Lipschitz constant can be invariant to the number of products $J$, from the assumption that $\sum_{j} (1 - s_j(\tilde{\beta})) = s_0(\tilde{\beta}) \leq c_2/J$. From this property, the variance term of \cite{berryLimitTheoremsEstimating2004} achieves $J$ rate, instead of $J^2$. Therefore, the rate of convergence may be improved with respect to the number of products $J$ relative to the one in this paper. Nevertheless, Ladeau-Talagrand contraction inequality does not apply with the $L_{\infty}$-Lipschitz case. While there is a recent study by \cite{fosterInftyVectorContraction2019} showing a tail probability bound for the Rademacher average with $L_{\infty}$-Lipschitz mapping, it is not trivial to apply to our case.
\end{remark}

Next we show the sufficient conditions for the empirical process at the true parameter value $\theta_0$ is bounded by the $log(JK)$ rate.
\begin{lemma}
 Assume that there is some $\sigma > 0$ such that
 \[
 \max_{j \in [J], l \in [2L], k \in [K]} \mathbb{E}_n (X_{jl} h_{jk}(W))^2 \leq \sigma^2
 \]
and
 \[
 \frac{\log(JK n)}{\sqrt{n}} \left(E[\|h_{jk}(W)\|^4_{\infty} \xi(\tilde{X};\theta_0)^4 \right)^{1/2} \leq \sigma^2
 \]
 Then 
 \[
 \|n^{-1/2} \mathbb{G}_n(g(\tilde{X},\theta_0))\|_{\infty} \leq n^{-1/2} C \sigma \sqrt{\log(JK)}
 \]
  with probability at least $1 - \delta /\log^2(n)$.
\end{lemma}
\begin{proof}
 The analogue argument in example 7 of Belloni et al (2018) in the application of lemma A.2 and A.3 shows the result.
\end{proof}

\begin{lemma} \label{lemma.LipGrad}
Assume the assumptions for lemma \ref{lemma.Lip}. Let $\nu_{j'g}$ and $\nu_{j'g0}$ as $JG$ vectors of indices as previously defined. Then there exists a sequence $B_{j'g}(x_l)$ which depends on the intermediate value of $\theta$ and $\theta_0$ such that
\[
 \frac{d \xi_j(\tilde{x};\theta)}{d\theta_l} - \frac{d \xi_j(\tilde{x};\theta_0)}{d\theta_l} = \sum_{j' \in [J],g \in [G]} B_{j'g}(x_l)(\nu_{j'g} - \nu_{j'g0}),
\]
and
\[
 |B_{j'g}(x_l)| \leq \bar{C}J\max_{j' \in [J]}|x_{j'l}|
\]
with a universal constant $\bar{C}$ for any value of $\theta$ and $x_l$.
\end{lemma} 
\begin{proof}
 Observe that
 \[
  \frac{d \xi_j(\tilde{x};\theta)}{d\theta_l} = D_{j}(\theta)\int\tilde{\beta}_{g(l)} s(\tilde{\beta};\tilde{x},\theta) \sum_{j' \in [J]}(1 - s_{j'}(\tilde{\beta};\tilde{x},\theta))x_{j'l} dF_{\tilde{\beta}}.
 \]
 
Below, we omit $\tilde{x}$ as the arguments of $s_{j}(\tilde{x};\theta)$ and $s_{j}(\tilde{\beta};\tilde{x},\theta)$ for notational simplicity.
 
Now we have
\[
 \left|\frac{d \xi_j(\tilde{x};\theta)}{d\theta_l} - \frac{d \xi_j(\tilde{x};\theta_0)}{d\theta_l}\right| = \left|\sum_{j',j'' \in [J]} \int \tilde{\beta}_{g(l)} \left[
 d_{jj''}(\theta) s_{j''}(\theta)(1 - s_{j'}(\theta)) - d_{jj''}(\theta_0) s_{j''}(\theta_0)(1 - s_{j'}(\theta_0))
 \right]
 x_{j'l} dF_{\tilde{\beta}}\right|.
\]
Thus,
\begin{align*}
   \sum_{j',j'' \in [J]} \int & \tilde{\beta}_{g(l)} \left[
 d_{jj''}(\theta) s_{j''}(\theta)(1 - s_{j'}(\theta)) - d_{jj''}(\theta_0) s_{j''}(\theta_0)(1 - s_{j'}(\theta_0))
 \right] x_{j'l} dF_{\tilde{\beta}}\\
  = \sum_{j',j'' \in [J]} \int & \tilde{\beta}_{g(l)}\left[ 
 d_{jj''} (\theta)(s_{j''}(\theta)(1 - s_{j'}(\theta)) - s_{j''}(\theta_0)(1 - s_{j'}(\theta_0)))
+ (d_{jj''}(\theta) - d_{jj''}(\theta_0))s_{j''}(\theta_0)(1 - s_{j'}(\theta_0))\right]x_{j'l} dF_{\tilde{\beta}}
\end{align*}

First consider expanding $s_{j''}(\theta)(1 - s_{j'}(\theta)) - s_{j''}(\theta_0)(1 - s_{j'}(\theta_0))$ with respect to $\nu_{jg} - \nu_{jg0}$. We have
\[
 \frac{d s_{j''}(\theta)(1 - s_{j'}(\theta))}{d\nu_{jg}} = \int \tilde{\beta}_g s_{j''}(\tilde{\beta};\theta)s_{j'}(\tilde{\beta};\theta)(1 - 2s_j(\tilde{\beta};\theta)) dF_{\tilde{\beta}}.
\]
Therefore, by the mean value theorem, there exists an intermediate value vector $\tilde{\theta}$,
\[
 s_{j''}(\theta)(1 - s_{j'}(\theta)) - s_{j''}(\theta_0)(1 - s_{j'}(\theta_0)) = \sum_{j=1}^J \sum_{g=0}^{G} \int \tilde{\beta}_g s_{j''}(\tilde{\beta};\tilde{\theta})s_{j'}(\tilde{\beta};\tilde{\theta})(1 - 2s_j(\tilde{\beta};\tilde{\theta})) dF_{\tilde{\beta}} (\nu_{jg} - \nu_{jg0}).
\]

Next consider expanding $d_{j'j''}(\theta) - d_{j'j''}(\theta_0)$ with respect to $\nu_{jg} - \nu_{jg0}$. Observe that
\begin{align*}
 \frac{dD^{-1}(\theta)}{d\nu_{jg}} =& -D^{-1}(\theta) \left[\frac{d}{d\nu_{jg}} \frac{ds_{j_1}(\theta)}{d\xi_{j_2}} \right]_{j_1,j_2} D^{-1}(\theta) \\
 =& -D^{-1}(\theta) \left[\frac{d}{d\nu_{jg}} s_{j_1}(\theta)(1 - s_{j_2}(\theta)) \right]_{j_1,j_2} D^{-1}(\theta) \\
 =& -D^{-1}(\theta) \left[\int \tilde{\beta}_g s_{j_1}(\tilde{\beta};\theta)s_{j_2}(\tilde{\beta};\theta)(1 - 2s_{j}(\tilde{\beta};\theta)) dF_{\tilde{\beta}}\right]_{j_1,j_2} D^{-1}(\theta)\\
 =& -\left[\sum_{j_1 \in [J]}\sum_{j_2 \in [J]} d_{j' j_2}(\theta)d_{j_1j''}(\theta) \int \tilde{\beta}_g s_{j_1}(\tilde{\beta};\theta)s_{j_2}(\tilde{\beta};\theta)(1 - 2s_{j}(\tilde{\beta};\theta)) dF_{\tilde{\beta}} \right]_{j',j''}.
\end{align*}
Therefore, by the mean value theorem
\begin{align*}
 d_{j'j''}(\theta) - d_{j'j''}(\theta_0) =& \sum_{j,j_1,j_2 \in [J]}
\sum_{g=0}^{G} d_{j'j_2}(\tilde{\theta}) d_{j_1 j''}(\tilde{\theta})\int \tilde{\beta}_g s_{j_1}(\tilde{\beta};\tilde{\theta})s_{j_2}(\tilde{\beta};\tilde{\theta})(1 - 2s_{j}(\tilde{\beta};\tilde{\theta})) dF_{\tilde{\beta}} (\nu_{jg} - \nu_{jg0}).
\end{align*}

Combining two results, we have
\begin{align*}
 \frac{d \xi_j(\tilde{x};\theta)}{d\theta_l} - \frac{d \xi_j(\tilde{x};\theta_0)}{d\theta_l}
 = \sum_{\tilde{j}=1}^J \sum_{g=0}^G B_{\tilde{j}g}(x_l) (\nu_{\tilde{j}g} - \nu_{\tilde{j}g0})
\end{align*}
where
\begin{align*}
 &B_{\tilde{j}g}(x_l) \equiv  
 \sum_{j', j'' \in [J]} \int \tilde{\beta}_{g(l)} d_{jj''}(\tilde{\theta}) \tilde{\beta}_{g} s_{j''}(\tilde{\beta};\tilde{\theta})s_{j'}(\tilde{\beta};\tilde{\theta})(1 - 2s_{\tilde{j}}(\tilde{\beta};\tilde{\theta})) dF_{\tilde{\beta}} x_{j'l}\\
 &+ \sum_{j',j'' \in [J]} \int \tilde{\beta}_{g(l)}\sum_{j_1, j_2 \in [J]} d_{j'j_2}(\tilde{\theta}) d_{j_1 j''}(\tilde{\theta})\tilde{\beta}_{g} s_{j_1}(\tilde{\beta};\tilde{\theta})s_{j_2}(\tilde{\beta};\tilde{\theta})(1 - 2s_{\tilde{j}}(\tilde{\beta};\tilde{\theta})) s_{j''}(\theta_0)(1 - s_{j'}(\theta_0)) dF_{\tilde{\beta}} x_{j'l}.
\end{align*}

For the second claim, observe that the absolute value of the first term of $B_{j'g}(x_l)$ is bounded above by
\[
 \frac{J}{C_1} \max_{j' \in [J]}|x_{j'l}| \int |\tilde{\beta}_{g(l)} \tilde{\beta}_{g}| \left|\sum_{j', j'' \in [J]} s_{j''}(\tilde{\beta};\tilde{\theta}) s_{j'}(\tilde{\beta};\tilde{\theta})  (1 - 2s_j(\tilde{\beta};\tilde{\theta}) )\right| dF_{\tilde{\beta}} \leq C_2 J \max_{j'}|x_{j'l}|.
\]
because the crude bound of $0 < s_j(\tilde{\beta};\tilde{\theta}) < 1$ says that
\begin{align*}
 &\left|\sum_{j', j'' \in [J]} s_{j''}(\tilde{\beta};\tilde{\theta}) s_{j'}(\tilde{\beta};\tilde{\theta}) (1 - 2s_j(\tilde{\beta};\tilde{\theta}))\right| \\
 &\leq |(1 - s_0(\tilde{\beta};\tilde{\theta}) - s_j(\tilde{\beta};\tilde{\theta}))^2 (1 - 2s_j(\tilde{\beta};\tilde{\theta})) + (1 - s_0(\tilde{\beta};\tilde{\theta}) - s_j(\tilde{\beta};\tilde{\theta}))(s_j(\tilde{\beta};\tilde{\theta})(1 - 2s_j(\tilde{\beta};\tilde{\theta})) + s_j(\tilde{\beta};\tilde{\theta})^2(1 - 2s_j(\tilde{\beta};\tilde{\theta}))|\\
 &\leq |1 - 2s_j(\tilde{\beta};\tilde{\theta})||1 + s_j(\tilde{\beta};\tilde{\theta}) + s_j(\tilde{\beta};\tilde{\theta})^2| \leq 3.    
\end{align*}

Similarly, the absolute value of the second term is bounded above by 
\begin{align*}
 &\frac{J^2}{C_3} \max_{j' \in [J]}|x_{j'l}| \int |\tilde{\beta}_{g(l)} \tilde{\beta}_{g}| \left|\sum_{j',j'',j_1,j_2 \in [J]}s_{j_1}(\tilde{\beta};\tilde{\theta})s_{j_2}(\tilde{\beta};\tilde{\theta})(1 - 2s_j(\tilde{\beta};\tilde{\theta}))s_{j''}(\tilde{\beta};\tilde{\theta})(1 - s_{j'}(\tilde{\beta};\tilde{\theta}))\right|dF_{\tilde{\beta}} \\
 &\leq \frac{J^2}{C_3} \max_{j' \in [J]}|x_{j'l}| \int |\tilde{\beta}_{g(l)} \tilde{\beta}_{g}| \left|s_0(\tilde{\beta};\tilde{\theta}) \sum_{j'',j_1,j_2 \in [J]}s_{j_1}(\tilde{\beta};\tilde{\theta})s_{j_2}(\tilde{\beta};\tilde{\theta})(1 - 2s_j(\tilde{\beta};\tilde{\theta}))s_{j''}(\tilde{\beta};\tilde{\theta})\right|dF_{\tilde{\beta}} \leq C_4 J \max_{j' \in [J]}|x_{j'l}|.
 \end{align*}
since $|d_{j_1 j_2}(\theta)d_{j_3 j_4}(\theta)| \leq |d_{j_1 j_2}(\theta)||d_{j_3 j_4}(\theta)|$. Therefore, the statement claimed follows for a constant $\bar{C} \geq \max\{C_2,C_4\}$.

\end{proof}

\begin{corollary} \label{crr.contraction}
 In addition to the assumptions for lemma \ref{lemma.Lip}, suppose that
 \begin{enumerate}
     \item 
        $\sup_{\Delta \theta \in \Theta, j,j' \in [J],k,k' \in [K]} \mathbb{E}_n Var((f_{jk}(\tilde{X};\theta_0 + \Delta \theta) - f_{jk}(\tilde{X};\theta_0)) f_{j'k'}(\tilde{X};\theta_0)) \leq B_{1n}^2$, and
    \item $\max_{j,j' \in [J], l \in [L], k,k' \in [K]} \mathbb{E}_n (X_{jl} h_{jk}(W)h_{j'k'}(W)\xi_{j}(\tilde{X};\theta_0))^2 \leq B_{2n}^2$ with probability at least $1 - \delta_n/6$,
 \end{enumerate}
 then,
 \begin{align*}
  \sup_{\theta \in \mathcal{R}(\theta_0), j,j' \in [J],k,k' \in [K]}& |\mathbb{G}_n (f_{j,k}(\tilde{X};\theta) - f_{j,k}(\tilde{X};\theta_0))f_{j',k'}(\tilde{X};\theta_0)|\\
  &\leq n^{-1/2} C(B_{1n} + (J^2G)(2\sqrt{2}B_{2n}\|\theta - \theta_0 \|_1 \log^{1/2}(8J^2 GKL/\delta_n)   
 \end{align*}
 with probability at least $1 - \delta_n$
 with a universal constant $C$.
  
\end{corollary}
\begin{proof}
 All the arguments in the proof of theorem 1 applies by replacing $h_{jk}(W_i)$ terms with $h_{jk}(W_i)h_{j'k'}(W_i)\xi_{j'k'}(\tilde{X};\theta_0)$. 
\end{proof}
 
\begin{corollary} \label{crr.contraction2}
 In addition to the assumptions for lemma \ref{lemma.LipGrad}, suppose that
 \begin{enumerate}
     \item 
        $\sup_{\Delta \theta \in \Theta, j \in [J],k \in [K], l \in [2L]} \mathbb{E}_n Var(G_{jk,l}(\tilde{X};\theta_0 + \Delta \theta) - G_{jk,l}(\tilde{X};\theta_0))\leq B_{1n}^2$, and
    \item $\max_{j \in [J], l \in [L], k \in [K]} \mathbb{E}_n (h_{jk}(W) X_{jl} \max_{j' \in [J]}|X_{j'l}| )^2 \leq B_{2n}^2$ with probability at least $1 - \delta_n/6$,
 \end{enumerate}
 then,
 \[
  \sup_{j,j' \in [J],k,k' \in [K], l \in [2L]} |\mathbb{G}_n (G_{jk,l}(\tilde{X};\hat{\theta}) - G_{jk,l}(\tilde{X};\theta_0))\|_{\infty} \leq n^{-1/2} C(B_{1n} + (J^2G)(2\sqrt{2}B_{2n}\|\hat{\theta} - \theta_0 \|_1 \log^{1/2}(8J^2 GKL/\delta_n)))
 \]
 with probability at least $1 - \delta_n$ with a universal constant $C$.
  
\end{corollary}
\begin{proof}
 By lemma \ref{lemma.LipGrad}, the gradient functions $G_{jk,l}(\tilde{X};\theta)$ can be linearly expanded with respect to $\nu_{jg} - \nu_{jg0}$ indices and their coefficients depend on $l$ only through the corresponding sub-vector of covariates $X_l$. Therefore, all the arguments in the proof of theorem 1 applies by replacing
 $f_{jk}(\tilde{X};\theta)$ terms with $G_{jk,l}(\tilde{X};\theta)$ and
 $h_{jk}(W_i)$ terms with $h_{jk}(W_i)\max_{j' \in [J]} |X_{j',l}|$. 
\end{proof}

\begin{corollary}(Based on \citet[theorem 4.12]{ledouxProbabilityBanachSpaces1991}) \label{crr.LT}
 Let $F:\mathbb{R}_+ \rightarrow \mathbb{R}_+$ be convex and increasing. Let $\mathcal{N}$ be a subset of $\mathbb{R}^{nJ}$ and $\mathcal{N}_i, \mathcal{N}_j$, and $\mathcal{N}_{ij}$ for each $i \in [n]$ and $j \in [J]$ be $i,j, (ij)$-th coordinates of $\mathcal{N}$. Let $\sigma = \{\sigma_i\}_{i \in n}$ be independent Rademacher random variables taking $\{-1,1\}$ with equal probability. Let $\phi_i : \mathcal{N}_i \rightarrow \mathbb{R}$ be functions such that $|\phi_i(\nu_i)| \leq 1$ and $|\phi_i(\nu_i)\nu_{ij} - \phi_i(\nu_i')\nu'_{ij}| \leq |\nu_{ij} - \nu'_{ij}|$
 uniformly over $\nu_i \in \mathcal{N}_i$, $\nu_{ij},\nu'_{ij} \in \mathcal{N}_{ij}$ for every $i \in [n]$ and $j \in [J]$. Then
 \[
 E\left[F\left(\frac{1}{2} \sup_{\nu \in \mathcal{N}} \left|\sum_{i=1}^n \sigma_i \phi_i(\nu_i)\nu_{ij}\right| \right) \right] \leq  E\left[F\left(\sup_{\nu_j \in \mathcal{N}_j} \left|\sum_{i=1}^n \sigma_i \nu_{ij}\right| \right) \right].
 \]
\end{corollary}
\begin{proof}
 The result follows from the proof of the \cite{ledouxProbabilityBanachSpaces1991}, Theorem 4.12. Below, we state a modified sketch of the original proof. First, we want to show that
 \[
 E\left[G\left( \sup_{\nu \in \mathcal{N}} \sum_{i=1}^n \sigma_i \phi_i(\nu_i)\nu_{ij}\right) \right] \leq  E\left[G\left(\sup_{\nu_j \in \mathcal{N}_j} \sum_{i=1}^n \sigma_i \nu_{ij} \right) \right]
 \]
 for convex and increasing $G:\mathbb{R} \rightarrow \mathbb{R}$. Once the above inequality holds, we would achieve the stated inequality by the symmetry of the distribution of the random variables multiplied with Rademacher variables.
 
 We show the above inequality by conditioning and iteration. 
 Let $\sigma_{i > j} \equiv \{\sigma_j,\ldots, \sigma_{n}\}$. Now, order the $2^{n-j}$ support values of $\sigma_{i > j}$. Let $\sigma_{i>j}^r$ be a $r$th value in the ordered support values of $\sigma_{i > j}$. As the Rademacher variables are independent, 
 \begin{align*}
 E\left[G\left( \sup_{\nu \in \mathcal{N}} \sum_{i=1}^n \sigma_i \phi_i(\nu_i)\nu_{ij} \right) \right] = & \sum_{r=1}^{2^{n-1}} E\left[G\left( \sup_{\nu \in \mathcal{N}} \sigma_1 \phi_1(\nu_1)\nu_{1j} + \sum_{i>1}^n \sigma_i^r \phi_i(\nu_i)\nu_{ij} \right) \middle|\sigma_{i > 1}^r \right]\left(\frac{1}{2}\right)^{n-1}\\
 = & \sum_{r=1}^{2^{n-1}} E\left[G\left( \sup_{\nu \in \mathcal{N}} \sigma_1 \phi_1(\nu_1)\nu_{1j} + \sum_{i>1}^n \sigma_i^r \phi_i(\nu_i)\nu_{ij} \right) \right]\left(\frac{1}{2}\right)^{n-1}.
 \end{align*}
 If
 \[
  E\left[G\left( \sup_{\nu_1 \in \mathcal{N}_1, t \in \mathbb{R}} \sigma_1 \phi(\nu_{1})\nu_{1j} + t\right)\right] \leq E\left[G\left( \sup_{\nu_{1j} \in \mathcal{N}_{1j}, t \in \mathbb{R}} \sigma_1 \nu_{1j} + t\right)\right]
 \]
 then we have
 \begin{align*}
  &\sum_{r=1}^{2^{n-1}} E\left[G\left( \sup_{\nu \in \mathcal{N}} \sigma_1 \phi_1(\nu_1)\nu_{1j} + \sum_{i>1}^n \sigma_i^r \phi_i(\nu_i)\nu_{ij} \right) \right]\left(\frac{1}{2}\right)^{n-1}\\
  &\leq \sum_{r=1}^{2^{n-1}} E\left[G\left( \sup_{\nu_{1j} \in \mathcal{N}_{1j}, \nu_{-1} \in \mathcal{N}_{-1}} \sigma_1 \nu_{1j} + \sum_{i>1}^n \sigma_i^r \phi_i(\nu_i)\nu_{ij} \right) \right]\left(\frac{1}{2}\right)^{n-1},     
 \end{align*}
 therefore, we achieve the target inequality by iterating over $r > 1$.
 
Now we show for all $t_1,s_1 \in \mathcal{N}_1$ and $t_2,s_2 \in \mathcal{N}_2$, 
 \[
  \frac{1}{2}G\left(s_{1j} - \phi(s_2)s_{2j} \right) + \frac{1}{2}G\left(t_{1j} + \phi(t_2)t_{2j}\right) \leq 
  \frac{1}{2}G\left(s_{1j} -s_{2j} \right) + \frac{1}{2}G\left(t_{1j} + t_{2j}\right).
 \]
 The remaining argument follows essentially the same argument of the proof of \cite{ledouxProbabilityBanachSpaces1991} but the fact that $\phi(s)$ takes a vector argument. Nevertheless, a similar argument applies because it is uniformly bounded by constant.
 First, we may assume that
 \[
  t_{1j} + \phi(t_2)t_{2j} \geq s_{1j} + \phi(s_2)s_{2j}
 \]
 and
 \[
  s_{1j} - \phi(s_2)s_{2j} \geq t_{1j} - \phi(t_2)t_{2j}
 \]
 otherwise the two separate supremum under $\sigma = 1$ and $\sigma = -1$ is solved as a single supremum under common variables either $(t_1,t_2)$ or $(s_1,s_2)$ only.
 We distinguish between the following cases. When $t_{2j} \geq s_{2j} \geq 0$, we have
 \begin{align*}
  t_{1j} + \phi(t_2) t_{2j} - s_{1j} + s_{2j}  \geq & s_{1j} + \phi(s_2)s_{2j} - s_{1j} + s_{2j}\\
  \geq & s_{2j} - |\phi(s_2)|s_{2j}\\
  = & (-|\phi(s_2)| + 1)s_{2j} \geq 0,     
 \end{align*}
 and
 \[
  s_{2j} - \phi(s_2)s_{2j} \leq t_{2j} - \phi(t_2)t_{2j}
 \]
 from $|\phi(t_2)t_{2j} - \phi(s_2)s_{2j}| \leq |t_{2j} - s_{2j}|$ and $t_{2j} \geq s_{2j}$. Therefore, we have
 \begin{align*}
  G(s_{1j}-\phi(s_2) s_{2j}) - G(s_{1j}-s_{2j}) \leq &G(s_{1j}-s_{2j} + (1 - \phi(s_2)) s_{2j}) - G(s_{1j}-s_{2j})\\
  \leq &G(t_{1j} + \phi(t_2) t_{2j}  + (1 - \phi(s_2)) s_{2j}) - G(t_{1j} + \phi(t_2) t_{2j})\\
  \leq &G(t_{1j} + s_{2j} - \phi(s_2) s_{2j} + \phi(t_2) t_{2j}) - G(t_{1j} + \phi(t_2) t_{2j})\\  
  \leq &G(t_{1j} + t_{2j} - \phi(t_2) t_{2j} + \phi(t_2) t_{2j}) - G(t_{1j} + \phi(t_2) t_{2j})\\    
  \leq &G(t_{1j} + t_{2j}) - G(t_{1j} + \phi(t_2) t_{2j})
 \end{align*}
  as $G(\cdot + x) - G(\cdot)$ is increasing for any $x \geq 0$. Thus, the desired inequality is achieved.
  
  The same argument applies with $t$ replaced with $s$ and $\phi$ into $-\phi$. The parallel argument holds when $t_{2j} \leq s_{2j} \leq 0$.
  
  When $t_{2j} \geq 0$ and $s_{2j} \leq 0$, 
  \[
   G(t_{1j} + \phi(t_2) t_{j2}) - G(t_{1j} + t_{2j}) \leq G(t_{1j} + |\phi(t_2)| t_{2j}) - G(t_{1j} + t_{2j}) \leq 0
  \]    
  and
  \[
   G(s_{1j}-\phi(s_2) s_{2j}) - G(s_{1j}-s_{2j}) \leq G(s_{1j} -|\phi(s_2)| s_{2j}) - G(s_{1j}-s_{2j}) \leq 0.
  \]    
  The parallel argument applies when $t_{2j} \leq 0$ and $s_{2j} \geq 0$.
\end{proof}

\end{appendix}

\end{document}